\documentclass[reqno,11pt,a4]{amsart}
\usepackage{amsmath,amssymb,tabularx,setspace}
\usepackage{color,multirow,graphics,graphicx}
\usepackage{appendix}
\newtheorem{theorem}{Theorem}
\numberwithin{equation}{section}
\numberwithin{theorem}{section}

\numberwithin{table}{section}

\newtheorem{definition}{Definition}
\newtheorem{remark}{Remark}
\newtheorem{corollary}{Corollary}

\numberwithin{remark}{section}
\numberwithin{corollary}{section}
\numberwithin{definition}{section}
\newcommand{\Var}{\mathrm{Var}}

\newcommand{\Cov}{\mathrm{Cov}}
\makeatletter
\renewcommand\section{\@startsection {section}{1}{\z@}%
                                   {-3.5ex \@plus -1ex \@minus -.2ex}%
                                   {2.3ex \@plus.2ex}%
                                   {\normalfont\large\bfseries}}
\makeatother

\newtheorem{Corollary}{Corollary}

\newtheorem{Theorem}{Theorem}

\newtheorem{Example}{Example}

\begin{document}
\doublespace
\vspace{-0.3in}
\title[]{Jackknife Empirical Likelihood  Ratio Test for independence between a continuous and a categorical random variable  }
\author[]%
{   S\lowercase{aparya} S\lowercase{uresh\textsuperscript{a} and } S\lowercase{udheesh} K. K\lowercase{attumannil\textsuperscript{b}  }
\\
 \lowercase{\textsuperscript{a}}I\lowercase{ndian} I\lowercase{nstitute of } M\lowercase{anagement}, K\lowercase{ozhikode}, I\lowercase{ndia},\\ \lowercase{\textsuperscript{b}}I\lowercase{ndian} S\lowercase{tatistical} I\lowercase{nstitute},
  C\lowercase{hennai}, I\lowercase{ndia}.}
\thanks{{$^{\dag}$}{Corresponding author E-mail: \tt skkattu@isichennai.res.in}}
\maketitle
\vspace{-0.2in}

\begin{abstract}
The categorical Gini covariance is a measure of dependence between a numerical variable and a categorical variable, quantifying the difference between conditional and unconditional distribution functions. The categorical Gini covariance equals zero if and only if the numerical variable and the categorical variable are independent. Inspired by this property, we propose a non-parametric test to assess the independence between a numerical and categorical variable using a modified version of the categorical Gini covariance. We used the theory of $U$-statistics to find the test statistics and study the properties. The proposed test has an asymptotic normal distribution under both the null and alternative hypotheses.   Since implementing a normal-based test is difficult,  we developed a jackknife empirical likelihood (JEL) ratio test for testing independence. Monte Carlo simulation studies are performed to validate the performance of the proposed JEL ratio test. We illustrate the test procedure using two real data sets.
\end{abstract}

\section{Introduction}\vspace{-.1in}
Measuring the strength of associations or correlations between two variables is important in many scenarios. Many studies look at various methods to capture the correlations for two numerical variables (Gibbons, 1993; Mari and Kotz, 2001). While the popular Pearson correlation measures the linear relationship, Spearman's $\rho$ and Kendall's $\tau$ explore monotonic relationships between two numerical variables. Similarly, different measures capture the associations between numerical or ordinal random variables (Shevlyakov and Smirnov, 2001). These correlation metrics, however, cannot be applied directly to a categorical variable. Moreover, for most of these proposed measures, while independence between $X$ and $Y$ implies a correlation of zero, the converse does not hold. Many approaches were suggested by various researchers, like Beknazaryan et al. (2019) and Dang et al. (2019), for measuring the dependence between categorical and continuous random variables. However, they lacked computational flexibility.

The categorical Gini correlation, introduced by Dang et al. (2021), is a measure of dependence between a continuous random variable $X$ and a categorical random variable $Y$. Consider a continuous random vector $X$ with a distribution $F$ in $\mathbb{R}^d$. Let $Y$ be a categorical random variable taking values $y_1,\cdots,y_K$ and its distribution given by $P_y$ such that $P(Y=y_k) =p_k >0$ for $k=1,2,\cdots,K$. Let $F_k$ be the conditional distribution of $X$, given that $Y=y_k$. When the conditional distribution of $X$ given $Y$ is the same as the marginal distribution of $X$, then $X$ and $Y$ are independent. Otherwise, they are dependent. The categorical Gini covariance and the correlation look at how much the marginal and conditional distributions differ from each other to measure dependence between $X$ and $Y$.

Let $\Psi_k$ and $\Psi$ be the characteristic functions of $F_k$ and $F$, respectively. We define the $L_2$ distance between $\Psi_k$ and $\Psi$ as
$$ T(F_k,F) = c(d) \int_{\mathbf{R}^d} \frac{|\Psi_k(t) -\Psi(t)|^2}{||t||^{d+1}} dt, $$
where $c(d) = \Gamma((d+1)/2) / \pi^{\frac{(d+1)}{2}}$. Then the Gini covariance between $X$ and $Y$ is defined as

$$gCov(X,Y)= \sum_{k=1}^{K} p_k T(F_k,F). $$

The Gini covariance measures the dependence of $X$ and $Y$ by quantifying the difference between the conditional and unconditional characteristic functions. The corresponding Gini correlations standardize the Gini covariance to have a range in $[0,1]$. When $d=1$, the categorical Gini covariance and correlation between $X$ and $Y$ can be defined by
\begin{equation}\label{gcov}
    gCov ( X,Y)= \sum_{k=1}^{K} p_k \int_{\mathbb{R}} (F_k(x)-F(x))^2 dx
\end{equation}
\begin{equation*}\label{gcor}
    \rho_g(X,Y)=  \frac{\sum_{k=1}^{K} p_k \int_{\mathbb{R}} (F_k(x)-F(x))^2 dx}{\int_{\mathbb{R} }(F(x) -F(x))^2dx}.
\end{equation*}
The Gini covariance is the weighted squared distance between the marginal and conditional distribution. It has been shown that $\rho_g$ has a lower computation cost. It is more straightforward to perform statistical inference and a more robust measure even with unbalanced data than the popular distance correlation measures (Dang et al., 2021). Hewage and Sang (2024) developed the jackknife empirical likelihood (JEL) method to identify the confidence interval of the categorical Gini correlation. Sang et al. (2021) demonstrated how to use the categorical Gini correlation to develop a test for the equality of $K$ distributions.

In this paper, we propose a non-parametric test to assess the independence between a numerical variable and a categorical variable, leveraging a modified categorical Gini covariance and the framework of $U$-statistics. Although other tests exist for assessing the independence of categorical and continuous random variables, this study represents the first attempt to utilize the categorical Gini correlation specifically for testing their independence. We organize the rest of the paper as follows: In Section 2, we develop a $U$-statistics-based test and study the asymptotic properties of the same. We also develop a jackknife empirical likelihood (JEL) ratio test for testing independence. Section 3 outlines the Monte Carlo simulations conducted to evaluate the performance of the proposed JEL ratio test. In Section 4, we illustrate the test procedure using two real data sets, and finally, we conclude in Section 5.
\vspace{-0.2in}
\section{Test Statistic}\vspace{-0.1in}
Consider a continuous random variable $X$ with a distribution function $F$ in $\mathbb{R}$. Denote $\Bar{F}(x)=1-F(x)$ as the survival function of $X$. Let $Y$ be a categorical random variable taking values $y_1,\cdots,y_K$ with probability mass function $p_k=P(Y=y_k),\,k=1,2,\cdots K$. Also, let $F_k$ be the conditional distribution of $X$ given $Y=y_k$, that is, $F_k(x) = P(X\leq x| Y=y_k)$. Denote $\bar{F}_k(x) = P(X> x| Y= y_k)$. Based on the observed data, we are interested in testing the null hypothesis: \vspace{-0.2in}
  $$H_0:X \text{ and }  Y \text{ are independent}$$
 against the alternative hypothesis\vspace{-0.2in}
 $$H_1:X \text{ and }  Y \text{ are dependent}.$$
It is important to note that the alternative hypothesis considered here is quite general, as $H_1$ only indicates a lack of independence and does not assume any specific form of dependence between $X$ and $Y$ (See, for example, Sudheesh et al. (2025)).

We define a departure measure $\Delta$ from the null hypothesis $H_0$ towards the alternative hypothesis $H_1$ as follows:\vspace{-0.1in}
\begin{equation}\label{delta}
   \Delta=  \sum_{k=1}^{K} p_k \int_{\mathbb{R}} (F_k(x)-F(x))^2 dF(x).
\end{equation}

The measure $\Delta$ is a weighted average of the squared distance between the marginal distribution and the conditional distribution. We call $\Delta$ defined in (\ref{delta}) as modified Gini covariance. When $X$ and $Y$ are independent, it is evident that $F_k$ and $F$ become equal, resulting in $\Delta=0$. Observe that each term of $\Delta$ is proportional to the $L^2$ distance between marginal and conditional distributions. Reformulating the expression for $\Delta$, it can be shown that $\Delta = 0$ implies the independence between $X$ and $Y$(Dang et al., 2021).

To find the test statistics, we simplify $\Delta$ as follows:
\begin{eqnarray}\nonumber \label{gnumer}
    \Delta &=& \sum_{k=1}^K p_k \int_{\mathbb{R}}(F_k(x)-F(x))^2 dF(x)\\ \nonumber
    &=& \sum_{k=1}^K p_k \int_{\mathbb{R}}\left(-\Bar{F}_k(x)+\Bar{F}(x)\right)^2 dF(x)\\\nonumber
    &=& \sum_{k=1}^K p_k \int_{\mathbb{R}}\left(\Bar{F}^2_k(x)+\Bar{F}(x)^2-2\Bar{F}_k(x)\Bar{F}(x)\right) dF(x)\\ \nonumber
&=& \sum_{k=1}^K p_k \int_{\mathbb{R}}\Bar{F}^2_k(x)dF(x)+ \sum_{k=1}^K p_k \int_{\mathbb{R}}\Bar{F}^2(x)dF(x)\\ \nonumber
&&-2 \sum_{k=1}^K p_k \int_{\mathbb{R}}\Bar{F}_k(x)\Bar{F}(x)dF(x)\\ \nonumber
&=& \sum_{k=1}^K p_k \int_{\mathbb{R}}\Bar{F}^2_k(x)dF(x)+ \sum_{k=1}^K p_k \int_{\mathbb{R}}\Bar{F}^2(x)dF(x)\\ \nonumber
&&-2  \int_{\mathbb{R}}\left(\sum_{k=1}^K p_k\Bar{F}_k(x)\right)\Bar{F}(x)dF(x)\\ \nonumber
&=& \sum_{k=1}^K p_k \int_{\mathbb{R}}\Bar{F}^2_k(x)dF(x)+ \sum_{k=1}^K p_k \int_{\mathbb{R}}\Bar{F}^2(x)dF(x)\\ \nonumber
&&-2  \int_{\mathbb{R}}\left(\sum_{k=1}^K P(X>x,Y=k)\right)\Bar{F}(x)dF(x). \nonumber
\end{eqnarray}

The $\Delta$ further simplifies as
\begin{eqnarray}\nonumber \label{gnumer}
    \Delta  &=& \sum_{k=1}^K p_k \int_{\mathbb{R}}\Bar{F}^2_k(x)dF(x)+ \int_{\mathbb{R}}\Bar{F}(x)^2dF(x)-2\int_{\mathbb{R}}\Bar{F}^2(x)dF(x)\\ \nonumber
 &=& \sum_{k=1}^{K}\frac{1}{p_k} \int_{\mathbb{R}}p_k^2\Bar{F}^2_k(x)dF(x)- \int_{\mathbb{R}}\Bar{F}^2(x)dF(x)\\
 &=& \sum_{k=1}^{K}\frac{1}{p_k}   \int_{\mathbb{R}} P(X>x, Y=k)^2 dF(x)- \frac{1}{3}\nonumber
 \\ &=& \sum_{k=1}^{K}\frac{1}{p_k}   \int_{\mathbb{R}} P(X_1>x, Y_1=k)P(X_2>x, Y_2=k) dF(x)- \frac{1}{3}\nonumber
 \\&=& \sum_{k=1}^{K}\frac{1}{p_k}   P(\min(X_1,X_2)>X_3, Y_1=Y_2=k)- \frac{1}{3}.
\end{eqnarray}

Note that if $X$ and $Y$ are independent, then equation \eqref{gnumer} becomes
\begin{eqnarray} \nonumber
   \Delta &=& \sum_{k=1}^{K}\frac{1}{p_k}   \int_{\mathbb{R}} P(X>x, Y=k)^2 dF(x)- \int_{\mathbb{R}}\Bar{F}^2(x)dF(x)\\ \nonumber
    &=& \sum_{k=1}^{K}\frac{1}{p_k}   \int_{\mathbb{R}}\left( p_k P(X>x)\right )^2 dF(x)- \int_{\mathbb{R}}\Bar{F}^2(x)dF(x)\\ \nonumber
    &=& \sum_{k=1}^{K} p_k   \int_{\mathbb{R}}\left(P(X>x)\right )^2 dF(x)- \int_{\mathbb{R}}\Bar{F}^2(x)dx\\ \nonumber
    &=&   \int_{\mathbb{R}}\left(P(X>x)\right )^2 dF(x)- \int_{\mathbb{R}}\Bar{F}^2(x)dF(x)=0.
\end{eqnarray}
Thus, under $H_0$, $\Delta$ is zero, and it is positive under the alternative $H_1$.

In (\ref{gnumer}), we represent $\Delta$ as a sum of probabilities. This enables us to use the theory of $U$-statistics to develop the test.
Let $(X_1,Y_1), (X_2,Y_2)$, and $(X_3,Y_3)$ be independent and identical copies of $(X,Y)$.
Define,
\begin{equation*}
    h_k^*((X_1,Y_1), (X_2,Y_2), (X_3,Y_3)) = I(\min(X_1,X_2) >X_3 , Y_1=Y_2 = k),
\end{equation*} where $I$ denotes the indicator function. Then $\Delta$ can be written as
\begin{equation*}
    \Delta = \sum_{k=1}^{K} \frac{1}{p_k} E(h_k^*((X_1,Y_1), (X_2,Y_2), (X_3,Y_3)))-\frac{1}{3}.
\end{equation*}
Let $(X_1,Y_1),\dots, (X_n,Y_n)$ be the observed data. An estimator of $\Delta$ is given by
\begin{small}
    \begin{eqnarray}\nonumber
        \widehat{\Delta}&=& \frac{6}{n(n-1)(n-2)}\sum_{k=1}^{K} \frac{1}{\widehat{p}_k}\sum_{i=1}^{n-2}\sum_{j=i+1}^{n-1}\sum_{l=j+1}^{n}h_k((X_i,Y_i), (X_j,Y_j), (X_l,Y_l))-\frac{1}{3},
    \end{eqnarray}
\end{small}where $h_k(.)$ is the symmetric version of the kernel $h_k^*(.)$ and $\widehat{p}_k$ is estimated by the following equation
$$\widehat{p}_k = \frac{1}{n} \sum_{i=1}^{n} I(Y_i =k).$$
    For convenience, we represent $\widehat \Delta$ as
    \begin{eqnarray}\nonumber
        \widehat{\Delta}&=& \sum_{k=1}^{K} \frac{1}{\widehat{p}_k}\frac{6}{n(n-1)(n-2)}\sum_{i=1}^{n-2}\sum_{j=i+1}^{n-1}\sum_{l=j+1}^{n}h_k((X_i,Y_i), (X_j,Y_j), (X_l,Y_l))-\frac{1}{3}\\
       &=& \sum_{k=1}^{K} \frac{1}{\widehat{p}_k}\widehat{\Delta}_k-\frac{1}{3}.
    \end{eqnarray}
The test procedure is to reject the null hypothesis $H_0$ against the alternative $H_1$ for a large value of $\widehat{\Delta}$.

Next, we study the asymptotic properties of the test statistic. By the law of large numbers, $\widehat{p}_k$ is a consistent estimator of $p_k$. The asymptotic properties of $U$-statistics ensure that $\widehat{\Delta}_k$ is a consistent estimator of $E(\widehat{\Delta}_k)$ (Sen, 1960). Hence, it is easy to verify that $\widehat{\Delta}$ is a consistent estimator of $\Delta$. Next, we determine the asymptotic distribution of $\widehat{\Delta}$.
 \begin{theorem}
     As $n\rightarrow \infty$, $\sqrt{n}(\widehat{\Delta}-\Delta)$ converges in distribution to a normal random variable with mean zero and variance $\sigma^2$ where
   $$\sigma^2 = \mathbf{a^T}\Sigma\mathbf{a},$$
    where $\mathbf{a}^T = (\frac{1}{p_1},\frac{1}{p_2},\cdots,\frac{1}{p_K})$ and $\Sigma = ((\sigma_{kl}))_{K\times K} $  with
    $$ \sigma_{kk}=9 \Var(E(h_k((X_1,Y_1), (X_2,Y_2), (X_3,Y_3)))|X_1,Y_1)),\, k=1,2,\cdots,K$$
     and
     \begin{eqnarray*}
         \sigma_{kl}&=&9 \Cov (h_k((X_1,Y_1), (X_2,Y_2), (X_3,Y_3)),h_l((X_1,Y_1), (X_4,Y_4), (X_5,Y_5))),\\&&\quad\,\text{ if }\, k \neq l,\,
     k,l =1,2,\cdots, K.
     \end{eqnarray*}
\end{theorem}
\begin{proof}Denote
$$\widehat{\Delta}^*=\sum_{k=1}^{K} \frac{1}{{p}_k}\widehat{\Delta}_k-\frac{1}{3}.$$
Observe that $\widehat{\Delta}_k$, $k=1,2,\ldots,K$, are $U$-statistics with kernel $h_k$ of degree 3. Also, $E(\widehat{\Delta}_k)=P(\min(X_1,X_2) >X_3, Y_1=Y_2 = k)=\Delta_k. $ Since $\widehat{\Delta}^*$ is a linear combination of $U$-statistics, $E(\widehat{\Delta}^*)=\Delta$.

Using the central limit theorem for  $U$-statistics (Lee, 2019), it follows that as $n \rightarrow \infty$, $\sqrt{n}(\widehat{\Delta}_k - \Delta_k)$ converges in distribution to a normal random variables with mean $0$ and variance $\sigma_{kk}$, where
\begin{equation}\label{var1}
    \sigma_{kk}=9 \Var(E(h_k((X_1,Y_1), (X_2,Y_2), (X_3,Y_3)))|X_1,Y_1)).
\end{equation}
By the weak law of large numbers, $\widehat p_k$ is a consistent estimator of $p_k$, and hence by Slutsky's theorem, the asymptotic distribution of  $\widehat{\Delta}^*$  is same as $\widehat{\Delta}$.
Hence, the asymptotic distribution of $\widehat{\mathbf{\Delta}}=(\sqrt{n}(\widehat{\Delta}_1 - {\Delta}_1),\ldots,\sqrt{n}(\widehat{\Delta}_K - {\Delta}_K))^T$ is multivariate normal with mean vector zero and variance-covariance matrix $\Sigma=((\sigma_{kl}))_{K\times K}$, where $\sigma_{kk}$ is given as in equation \eqref{var1} and $\sigma_{kl}$ is given by
\begin{small}
  \begin{eqnarray} \label{eq:sigkl}
    \sigma_{kl}=9 \Cov (h_k((X_1,Y_1), (X_2,Y_2), (X_3,Y_3)),h_l((X_1,Y_1), (X_4,Y_4), (X_5,Y_5))).
\end{eqnarray}
\end{small}
If a random vector $\mathbf{Z}$ follows a $p$-variate normal distribution with vector $\mathbf{\mu}$ and variance-covariance matrix $\Sigma$, then for any vector $\mathbf{a} \in \mathbb{R}^{p}$, the asymptotic distribution of $\mathbf{a}^{T}\mathbf{Z}$ is a univariate normal with mean $\mathbf{a}^{T}\mathbf{\mu}$ and variance $\mathbf{a}^T \Sigma \mathbf{a}$. Therefore, with the selection of $\mathbf{a}=(\frac{1}{p_1},\frac{1}{p_2},\cdots,\frac{1}{p_K})^T$, we have the desired result.
\end{proof}

Under the null hypothesis $H_0$, $\Delta=0 $ and $\sigma_{kk}$ and $\sigma_{kl}$ given in equations \eqref{var1} and \eqref{eq:sigkl} can be expressed as
\begin{equation} \label{eq:h0skk}
    \sigma_{kk}=\frac{p_k^3}{2}- \frac{31}{45}p_k^4
\end{equation}
and
\begin{equation} \label{eq:h0sigkl}
\sigma_{kl}=-\frac{31}{45}p_k^2 p_l^2.
\end{equation}
The proof for obtaining the expressions (\ref{eq:h0skk}) and (\ref{eq:h0sigkl}) is given in the appendix.
\begin{corollary}\label{cor1}
 Under $H_0$, as $n\rightarrow \infty$, $\sqrt{n}\widehat{\Delta}$ converges in distribution to a normal random variable with mean zero and variance $\sigma_0^2= \mathbf{a^T}\Sigma\mathbf{a}, $ where $\mathbf{a}^T = (\frac{1}{p_1},\frac{1}{p_2},\cdots,\frac{1}{p_K}) $ and $\Sigma = ((\sigma_{kl}))_{K\times K} $ with $ \sigma_{kk} $  and    $\sigma_{kl}$ are given as in equations \eqref{eq:h0skk} and \eqref{eq:h0sigkl}, respectively.
\end{corollary}
Using Corollary \ref{cor1}, we can obtain a normal-based critical region for the proposed test. Let $\widehat{\sigma}_0^2$ be a consistent estimator of $\sigma_0^2$. Then, we reject $H_0$ in favor of $H_1$ at significance level $\alpha$ if

$$\frac{\sqrt{n}\widehat{\Delta}}{\widehat{\sigma}_0}>Z_{\alpha},$$ where $Z_{\alpha}$ is the upper $\alpha$ percentile point of a standard normal distribution.

When the distributions of random variables are unknown, it is challenging to determine a consistent estimator for the null variance. $\sigma_0^2= \mathbf{a^T}\Sigma\mathbf{a},$ particularly when the number of categories is large. Hence, it is difficult to determine the critical region of the proposed test based on the asymptotic distribution of the test statistic. Consequently, a normal-based test is not recommended when the category variable involves many categories. Therefore, we propose a distribution-free test based on jackknife empirical likelihood to evaluate the proposed hypothesis of independence between $X$ and $Y$.
\vspace{-0.1in}
\subsection{JEL ratio test}
The concept of empirical likelihood was first proposed by Thomas and Grunkemier (1975) to determine the confidence interval for survival probabilities under right censoring. The work by Owen (1988, 1990) extended this concept of empirical likelihood to a general methodology. Empirical likelihood is calculated by maximizing the non-parametric likelihood function subject to certain constraints. When the constraints become nonlinear, evaluating the likelihood function becomes computationally challenging. Further, it becomes increasingly difficult as $n$ grows large. To overcome this problem, Jing et al. (2009) introduced the jackknife empirical likelihood method for finding the confidence interval of the desired parametric function. This method is highly popular among researchers because it combines the effectiveness of the likelihood approach with the jackknife technique.

Let $(X_1,Y_1),(X_2,Y_2), \cdots, (X_n,Y_n)$ be a random sample of size $n$ from a joint distribution function $F(x,y)$ of $X$ and $Y$. Recall the test statistic $\widehat{\Delta}$, given by \vspace{-0.1in}
\begin{eqnarray}\nonumber
        \widehat{\Delta}&=&\frac{6}{n(n-1)(n-2)}\sum_{k=1}^{K} \frac{1}{\widehat p_k}\sum_{i=1}^{n-2}\sum_{j=i+1}^{n-1}\sum_{l=j+1}^{n}h_k(X_i,X_j,X_l)-\frac{1}{3}\\ \nonumber
        &=& \sum_{k=1}^{K} \frac{1}{\widehat{p}_k}\widehat{\Delta}_{k}-\frac{1}{3}.
        \end{eqnarray}

 Let $\nu_i$ be the jackknife pseudo-values defined as follows:
$$\nu_i= n\widehat{\Delta}-(n-1) \widehat{\Delta}_{i}\,   i =1,\cdots,n,$$
 where $\widehat{\Delta}_{i}$ is the value of $\widehat{\Delta}$ obtained using $(n-1)$ observations $(X_1,Y_1),\cdots,$ $(X_{i-1},Y_{i-1}),(X_{i+1},Y_{i+1}), \cdots,(X_n,Y_n)$. And, we have $\widehat{\Delta}_{i} =\sum_{k=1}^{K} \frac{1}{\widehat{p}_k}\widehat{\Delta}^{i}_{k}-\frac{1}{3}$ where each $\widehat{\Delta}^{i}_{k}$ is calculated from the corresponding $n-1$ samples.

 Let ${q_i}, i=1,2,\ldots, n$ be the probability associated with each  $\nu_i$. Define JEL for departure measure $\Delta$ as
  \begin{equation*}
 JEL(\Delta)=\sup_{\bf q} \big(\prod_{i=1}^{n}{q_i};\,\, \sum_{i=1}^{n}{q}=1;\,\,\sum_{i=1}^{n}{q_i \nu_i}=0\big).
\end{equation*}
Note that $\prod_{i=1}^{n}{q_i}$ is maximized subject to the condition  $\sum_{i=1}^{n}{q_i}=1$ at $q_i = 1/n$.
Hence, using the Lagrange multiplier method, we obtain the jackknife empirical log-likelihood ratio as
\begin{equation}\label{likratio}
  \log R(\Delta)=-\sum_{i=1}^{n}\log\big(1+\lambda \nu_i\big),
\end{equation}
where $\lambda$ is the solution of
\begin{equation}\label{lambda2}
  \frac{1}{n}\sum_{i=1}^{n}{\frac{\nu_i}{1+\lambda \nu_i}}=0,
\end{equation}
    provided
\begin{equation*}
  \min_{{1\le i\le n}}\nu_i<\Delta<  \max_{1\le i \le n}\nu_i.
\end{equation*}
 Next, we obtain the asymptotic distribution of the jackknife empirical log-likelihood ratio as an analog of Wilk's theorem.
\begin{theorem} \label{thm_jel}
 Denote

 $$g(x,y)= E \left[ \sum_{k=1}^{K} \frac{1}{p_k}h_k((X_1,Y_1), (X_2,Y_2), (X_3,Y_3)) -\frac{1}{3}|X_1 =x ,Y_1=y\right]$$ and $\sigma_g^2= Var(g(X,Y))$. Suppose that  $E \left[ g(X,Y) \right] <\infty$ and $\sigma_{g}^{2}>0$. Under $H_0$, as $n\rightarrow \infty$, $-2\log R(0)$ converges in distribution to a $\chi^2$ random variable  with one degree of freedom.
\end{theorem}

\begin{proof}
   Consider
\begin{eqnarray} \nonumber
    \sum_{i=1}^{n} \nu_{i} &= &\sum_{i=1}^{n} \left[n \left(\sum_{k=1}^{K}\frac{1}{\widehat{p}_k}\widehat{\Delta}_k-\frac{1}{3}\right) -(n-1) \left(\sum_{k=1}^{K}\frac{1}{\widehat{p}_k}\widehat{\Delta}^{i}_{k}-\frac{1}{3}\right)\right] \\ \label{eqn:jelpf1}
    &=& \sum_{i=1}^{n} \sum_{k=1}^{K} \frac{1}{\widehat{p_k}}\left(n \widehat{\Delta}_k - (n-1) \widehat{\Delta}^{i}_{k}\right) -n \frac{1}{3}.
\end{eqnarray}
Note that $\sum_{i=1}^{n} \widehat{\Delta}^{i}_{k} = n \widehat{\Delta}_k$. Thus, the equation \eqref{eqn:jelpf1} becomes
\begin{eqnarray} \nonumber
     \sum_{i=1}^{n} \nu_{i} &=& \sum_{i=1}^{n} \sum_{k=1}^{K} \frac{1}{\widehat{p_k}}\left(n \widehat{\Delta}_k - (n-1) \widehat{\Delta}^{i}_{k}\right) -n \frac{1}{3}\\ \nonumber
     &=&  \left(\sum_{k=1}^{K}\frac{1}{\widehat{p}_k}n\widehat{\Delta}_k-n\frac{1}{3}\right)  \\ \nonumber
     &=& n \widehat{\Delta}.
\end{eqnarray}
Hence, we have
$$\widehat{\Delta} = \frac{1}{n}\sum_{i=1}^{n} \nu_i. $$


By Corollary \ref{cor1}, under $H_0$, we have
\begin{equation} \label{eq:jel_l1}
    \frac{\sqrt{n}\widehat{\Delta}}{2 \sigma_0} \rightarrow N(0,1).
\end{equation}
Further, denote $S= \frac{1}{n} \sum_{i=1}^{n} \nu_i^2 $ and by law of large number
\begin{equation} \label{eq:jel_l2}
S = 4\sigma_0^2 + O(1).
\end{equation}In view of the results stated in (\ref{eq:jel_l1}) and (\ref{eq:jel_l2}),  by Slutsky's theorem, we have
\begin{equation} \label{eq:jel_l3}
    \frac{\sqrt{n}\widehat{\Delta}}{\sqrt{S}} \rightarrow N(0,1).
\end{equation}Denote
\[
g_n(\lambda)
= \frac{1}{n}\sum_{i=1}^n \frac{\nu_i}{1+\lambda^\top \nu_i} = 0.
\]

Expanding $g_n(\lambda)$ around $\lambda=0$ gives
\[
0 = g_n(0) - \Bigg( \frac{1}{n}\sum_{i=1}^n \nu_i^2 \Bigg)\lambda + r_n,
\]
where
\[
g_n(0) = \widehat{\Delta},
\qquad
r_n = o_p(n^{-1/2}).
\]
Hence
\begin{equation}\label{lambda}
  \lambda
= S^{-1}\,\widehat{\Delta} + o_p(n^{-1/2}).
\end{equation}
Now, consider
\[
\begin{aligned}
-2\log R(0)
&= 2 \sum_{i=1}^n \log\!\bigl(1+\lambda \nu_i\bigr) \\[6pt]
&= 2 \sum_{i=1}^n \left( \lambda \nu_i - \tfrac{1}{2} (\lambda \nu_i)^2 \right) + o_p(1) \\[6pt]
&= 2\,\lambda \sum_{i=1}^n \nu_i \;-\; \lambda^2 \sum_{i=1}^n \nu_i^2 + o_p(1) \\[6pt]
&= 2n\,\lambda \,\widehat{\Delta} \;-\; n\,\lambda^2 S^2 + o_p(1) .
\end{aligned}
\]
Using the expression (\ref{lambda}), the above given log jackknife empirical likelihood ratio  becomes
$$ -2 \log R(0) = \frac{n \widehat{\Delta}  ^2}{S} + o_p(1).$$
In the view of equation \eqref{eq:jel_l3}, as $n \rightarrow \infty$, we have
$$-2 \log R(0)  \xrightarrow{d} \chi^2_1.$$
\end{proof}
Using Theorem \ref{thm_jel}, we can determine the critical region of the  JEL ratio test. We  reject the null hypothesis $H_0$ against the alternative hypothesis $H_1$ at a significance level $\alpha$ if
 \begin{equation}
 	-2 \log R(0)> \chi^2_{1,\alpha},
 \end{equation}	
where $\chi^2_{1,\alpha}$ is the upper $\alpha$ percentile point of the $\chi^2$ distribution with one degree of freedom.

\section{Simulation Study}
In this section, we carried out a Monte Carlo simulation study to assess the finite sample performance of the proposed JEL ratio test. The simulation is done ten thousand times using different sample sizes: 20, 40, 60, 80, and 100. Simulation is done using R software.

We carried out the simulation study in two parts. First, we determine the empirical type I error by considering a continuous random variable $X$ following a lognormal distribution with parameters $(\mu,\sigma)$ and $Y$, a categorical random variable with a discrete uniform distribution having six categories $(K=6)$. We independently generated the continuous random variable $X$ and the categorical random variable $Y$ to determine the empirical type I error of the JEL ratio test. In this case, we examined two different scenarios by varying the parameters of the assumed distribution of $X$. The results from the simulation study are presented in Table \ref{tab:typ1}. Additionally, we compared the empirical Type 1 error of the proposed test with some other tests: the mean-variance (MV) index test (Cui and Zhong, 2019), the integral Pearson chi-square (IPC) test (Ma et al., 2023), and the semi-distance (SD) test (Zhong et al., 2024). The MV and SD tests are available in the R package `semidist'. In the comparison, we assume that $X$ follows a lognormal distribution with parameters $(\mu,\sigma) = (0,1)$, while the categorical variable $Y$ follows a uniform distribution with six categories $(K=6)$. The comparison of empirical type I error is presented in Table \ref{tab:typ1com}. Table \ref{tab:typ1com} shows that as the sample size $n$ increases, the empirical type I error for all the tests approaches the significance level $\alpha$.
\begin{table}[]
\caption{Empirical type I error of JEL ratio test at 5\% significance level: $X$ has log normal distribution $(\mu,\sigma)$. }
    \centering
    \begin{tabular} {cccccccccc}
    \hline
   $n$ & $(\mu,\sigma) =(0,1)$ & $(\mu,\sigma) =(0,2)$\\
\hline

   20 & 0.020 &0.083\\
    40 & 0.041  &0.071\\
    60 & 0.065  &0.064\\
    80 & 0.061 &0.061\\
    100 & 0.051  &0.052\\
    \hline
    \end{tabular}
    \label{tab:typ1}
\end{table}

\begin{table}[]
\caption{Empirical type 1 error of the tests at 5\%  significance level.}
    \centering
    \begin{tabular} {cccccccccc}
    \hline
   $n$ & JEL  & IPC & MV & SD  \\
\hline

   20 & 0.031 &0.075&0.060&0.040 \\
    40 & 0.041  &0.055& 0.040&0.045\\
    60 & 0.043  &0.060&0.031&0.060\\
    80 & 0.059 &0.062&0.040&0.060\\
    100 & 0.051  &0.061&0.045&0.053\\

    \hline
    \end{tabular}

    \label{tab:typ1com}
\end{table}


In the second part of the simulation, we evaluate the empirical power of the JEL ratio test and compare it to the MV, IPC, and SD tests. In the simulation, we generate random samples from mixtures of normal, exponential, and lognormal distributions. We consider $K = 3$, as in Hewage and Sang (2024). Further, for each scenario, $\mathbf{p} =(p_1,p_2,p_3) $ is varied from balanced $\mathbf{p}=(1/3,1/3,1/3)$ to lightly unbalanced $\mathbf{p}=(5/12,4/12,3/12)$ to highly unbalanced $\mathbf{p}=(6/10,3/10,1/10)$. Therefore, we consider the following three scenarios:
\begin{itemize}
    \item Scenario 1: Balanced
    $$X \sim \frac{1}{3}\, \mathcal{N}(0,1) +\frac{1}{3}\, Exp(1) +\frac{1}{3}\, Lognormal(0,1)$$

\item Scenario 2: Lightly unbalanced
$$X \sim \frac{5}{12}\, \mathcal{N}(0,1) + \frac{4}{12}\, Exp(1) + \frac{3}{12} \,Lognormal(0,1)$$

\item Scenario 3: Highly unbalanced
    $$X \sim \frac{6}{10}\, \mathcal{N}(0,1) + \frac{3}{10}\, Exp(1) +\frac{1}{10}\, Lognormal(0,1).$$

\end{itemize}
The results for the power comparison are given in Tables \ref{tab:sc1}, \ref{tab:sc2}, and \ref{tab:sc3}. From these tables, we see that, in all scenarios, the empirical power of the test converges to one as the sample size increases. From Table \ref{tab:sc1}, we can observe that for Scenario 1, the power is high for even small samples $(n=20)$. We also observe that the proposed test performs better than the IPC, MV, and SD tests even when the sample size is smaller. In Scenario 2 given in Table \ref{tab:sc2}, we can see that SD test is performing better for small sample size. However, the power is comparable for larger sample size. In Scenario 3 as given in Table \ref{tab:sc3}, we can observe that our test performs better than all the other test when the sample is 40 or more.

\begin{table}[]
\caption{Empirical power of the tests at 5\%  significance level in Scenario 1}
    \centering
    \begin{tabular} {cccccccccc}
    \hline
   $n$ & Proposed test  & IPC & MV & SD  \\
\hline

   20 & 1.000 &0.530&0.015&0.763 \\
    40 & 1.000  &0.751& 0.251&0.971\\
    60 & 1.000  &0.989&0.486&0.998\\
    80 & 1.000 &0.999&0.733&1.000\\
    100 & 1.000  &1.000&0.881&1.000\\

    \hline
    \end{tabular}

    \label{tab:sc1}
\end{table}

\begin{table}[]
\caption{Empirical power of the tests at 5\%  significance level in Scenario 2}
    \centering
    \begin{tabular} {cccccccccc}
    \hline
   $n$ & Proposed test  & IPC & MV & SD \\
\hline

   20 & 0.511&0.502&0.012& 0.706\\
    40 & 0.781&0.886&0.182&0.952\\
    60 & 0.987&0.986&0.456&0.996\\
    80 & 1.000 &0.999&0.781&1.000\\
    100 & 1.000  &1.000&0.866&1.000\\

    \hline
    \end{tabular}

    \label{tab:sc2}
\end{table}

\begin{table}[]
\caption{Empirical power of the tests at 5\%  significance level in Scenario 3}
    \centering
    \begin{tabular} {cccccccccc}
    \hline
   $n$  & Proposed test & IPC & MV & SD \\
\hline

   20 & 0.565&0.453&0.029 & 0.568 \\
    40 & 0.980&0.756&0.0631 & 0.868\\
    60 & 0.999&0.950&0.234& 0.977\\
    80 & 1.000 &0.994&0.469&0.999\\
    100 & 1.000  &0.999&0.689&1.000\\

    \hline
    \end{tabular}

    \label{tab:sc3}
\end{table}

 Next, we report the computation time required to calculate the value of the test statistics, and it is given in Table \ref{tab:comp}. The simulation uses a laptop with an Intel (R) Core(TM) i7-8550U CPU \@ 1.80 GHz with 16 GB RAM. It is evident that the proposed methods offer several benefits but also have certain trade-offs. Here we observe that when the sample size is large, JEL-based methods demand more computational time than the other tests, as they involve resampling techniques. However, since the test is distribution-free, the proposed test provides greater flexibility, robustness, and inferential accuracy, which is demonstrated in terms of its higher power while dealing with different scenarios. Also, note that in small sample sizes, the computational times of all tests considered here are comparable.

\begin{table}
	\centering
	\caption{Time taken (in second)  for finding the values of different test statistics}
	\begin{tabular}{lccccccccc} \hline
		  $n$ & Proposed test  & IPC & MV & SD\\ \hline
20&0.041&0.025&0.010 &0.009\\	
40&0.051&0.028&0.021&0.016\\
60&2.362&0.034&0.032&0.028\\
80&6.968&0.048&0.043&0.035\\
100&16.382&0.059&0.055&0.042\\
\hline
\end{tabular}
    \label{tab:comp}
\end{table}

\section{Data Analysis}
To illustrate the application of the proposed JEL ratio test, we apply it to two real data sets. The first dataset used for illustration is the well-known Iris built-in data set in R. This data set comprises measurements in centimeters of sepal length, sepal width, petal length, and petal width for 150 iris flowers, with 50 samples each from three species: setosa, versicolor, and virginica. In this study, we focus on the sepal length measurements from the Iris data set to illustrate the application of the proposed test.

It is reasonable to assume that a flower's species affects sepal length. To provide clarity, Figure \ref{fig:den} presents density plots of sepal length for each species. Figure \ref{fig:den} illustrates that each species exhibits a distinct distribution of sepal lengths, suggesting a potential dependence of sepal length on species. To validate this inference, we applied the proposed hypothesis test. The calculated value of the JEL ratio is $-2 \log R(\Delta) = 6.70$ with a p-value of 0.013.  The result suggests that we can reject the null hypothesis at 5\% significance level. Thus, as per the test, we can conclude that the sepal length of the iris flower is dependent on species. The R code used for the data analysis is provided as supplementary material.

\begin{figure}[htbp]
\centering
\caption{Density plot of Sepal Length}
\vspace{0.2in}

\includegraphics[width=10 cm]{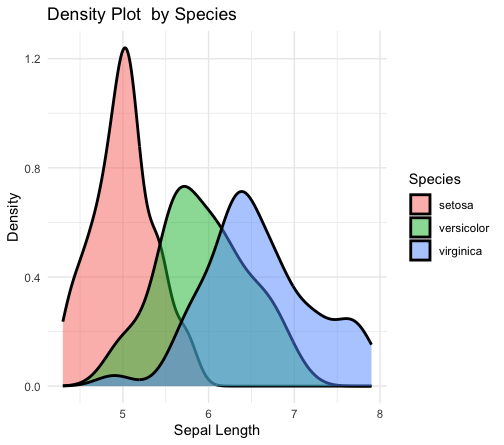}

\label{fig:den}

\end{figure}

The second data set analyzed is the gilgai survey data (Jiang et al., 2011). The dataset comprises 365 samples obtained from depths of 0-10 cm, 30-40 cm, and 80-90 cm beneath the Earth's surface. The dataset has three variables: pH, EC (expressed in mS/cm), and CC (expressed in ppm), all derived from a 1:5 soil-to-water extract of each sample. In the analysis, we use the depths as the categorical random variable $Y$, which takes on three values: 1, 2, and 3, corresponding to depths of 0-10 cm, 30-40 cm, and 80-90 cm, respectively. We regard the pH value as the continuous random variable $X$. Figure \ref{fig:deng} illustrates the density curves for pH values at three different depth levels. Figure \ref{fig:deng} indicates a possible association between the sample's pH level and depth. To validate this observation, we implemented the recommended test on the data. The computed value of the JEL ratio is $-2 \log R(\Delta) = 101.61$ with $p-$ value 0.00. The findings suggest that we can dismiss the null hypothesis at both the 5\% and 1\% significance levels. Consequently, we can ascertain that the pH level of the soil sample is dependent upon the depth.

\begin{figure}[htbp]
\centering
\caption{Density plot of pH level of the Soil}
\vspace{0.2in}

\includegraphics[width=10 cm]{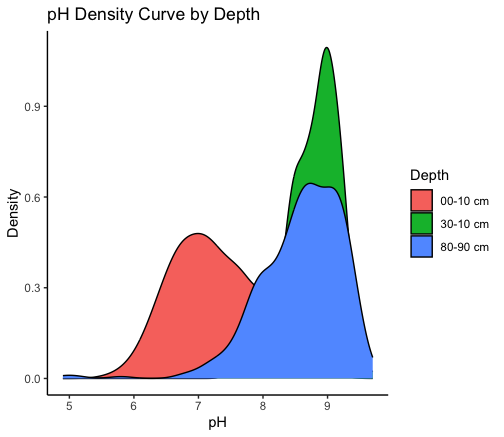}

\label{fig:deng}

\end{figure}

\section{Conclusion}
  Testing the independence between a continuous variable ($X$) and a discrete variable ($Y$) has many practical applications in various fields, such as economics, social sciences, environmental science, medical science, etc. Testing independence between continuous and discrete variables helps simplify models, improve accuracy, and ensure better interpretability in machine learning applications. For instance, a Gaussian mixture model can effectively model data when a continuous variable exhibits different distributions among different categories. Therefore, it is desirable to develop a test for independence when $X$ is continuous and $Y$ is discrete.

  Motivated by this, we proposed a $U$-statistics-based test to assess the independence between a categorical variable and a continuous variable. Given the difficulty of implementing the $U$-statistics test, we developed a JEL ratio test as an alternative. We compared our testing methodology with recently developed tests for independence, and the results indicate that it is a viable competitor to the existing methods. We demonstrated the testing methodology on two data sets, including the well-known Iris data set. The analysis of the Iris dataset revealed that the test identifies a dependence between the continuous variable, sepal length, and the categorical variable, flower species, highlighting its potential application in classification problems. This proposed test can be used for variable selection in discriminant analysis or generalized linear modeling (Cui et al., 2015).

 Various forms of censoring and truncation may occur in lifetime data analysis. A test for independence between a categorical and a continuous variable may be developed if a continuous variable is subjected to right censoring. The $U$-statistics for censored data can be employed to modify the proposed test procedure (Datta et al., 2010). Another possible direction for future research is to test independence in cases of left-truncated and right-censored data. (Sudheesh et al. 2023). Recently, Fan et al. (2024) proposes a test for independence between  two continuous random variables $X$ and $Y$  when there is measurement error in the variable $X$. A possible extension in our work can explored in the same direction as in Fan et al. (2024).

\section*{Acknowledgements}We thank the anonymous reviewers for their constructive suggestions, which helped us to improve the earlier version of this manuscript.


\section*{Appendix}
\subsection*{ Derivation of the variance term}
\begin{eqnarray}\nonumber
     \sigma_{kk}&=& \Cov\Big[\big(I(\min(X_1,X_2)>X_3, Y_1=Y_2=k)+I(\min(X_2,X_3)>X_1, Y_2=Y_3=k)\\ \nonumber
&+&I(\min(X_1,X_3)>X_2, Y_1=Y_3=k)\big),\big(I(\min(X_1,X_4)>X_5, Y_1=Y_4=k)\\ \nonumber
&+&I(\min(X_4,X_5)>X_1, Y_4=Y_5=k)+I(\min(X_1,X_5)>X_4, Y_1=Y_5=k) \big)\Big]\\ \nonumber
 &=& 3 \Cov[I(\min(X_1,X_2)>X_3, Y_1=Y_2=k),I(\min(X_1,X_4)>X_5, Y_1=Y_4=k)]\\ \nonumber
 &+& 5\Cov[I(\min(X_1,X_2)>X_3, Y_1=Y_2=k),I(\min(X_4,X_5)>X_1, Y_4=Y_5=k)]\\ \nonumber
 &+& \Cov[I(\min(X_2,X_3)>X_1, Y_2=Y_3=k),I(\min(X_4,X_5)>X_1, Y_4=Y_5=k)]\\ \label{sigkk}
 &=& 3A+5B+C.
\end{eqnarray}
 Consider the term $A$ given in equation \eqref{sigkk}.
 \begin{eqnarray}\nonumber
     A &=&  \Cov[I(\min(X_1,X_2)>X_3, Y_1=Y_2=k),I(\min(X_1,X_4)>X_5, Y_1=Y_4=k)]\\ \nonumber
     &=& E[I(\min(X_1,X_2)>X_3, Y_1=Y_2=k).I(\min(X_1,X_4)>X_5, Y_1=Y_4=k)] \\ \nonumber
     &-& P(\min (X_1,X_2)>X_3, Y_1=Y_2 =k) \times P(\min(X_1,X_4)>X_5, Y_1=Y_4=k)\\ \nonumber
     &=& a_1 - a_2.
 \end{eqnarray}
Under the null hypothesis we have,
 \begin{eqnarray} \nonumber
     a_1 &=& E(I(\min(X_1,X_2)>X_3, Y_1= k,Y_2=k).I(\min(X_1,X_4)>X_5, Y_1=k, Y_4=k)) \nonumber\\&=& E(I(\min(X_1,X_2)>X_3, Y_1 = k,Y_2=Y_4=k, \min(X_1,X_4)>X_5) )\nonumber\\&=&\nonumber \sum_{y=1}^K   \int_{0}^{\infty}  E(I(\min(x,X_2)>X_3, y = k,Y_2=Y_4=k,\min(x,X_4)>X_5))\\\nonumber
    &&\quad\quad\quad\quad\quad\quad\quad\quad\quad\quad\quad\quad\quad\quad\quad P(Y_1=y)dF(x).\nonumber\\\nonumber
    \end{eqnarray}Hence
     \begin{eqnarray} \nonumber
     a_1
    &=& \sum_{y=1}^K  \int_{0}^{\infty}  I(y=k) p_k^2 P(\min(x,X_2)>X_3)^2  P(Y_1=y)dF(x)\\ \nonumber
    &=& p_k^3  \int_{0}^{\infty}  P(\min(x,X_2)>X_3)^2 dF(x) \\ \nonumber
    &=& p_k^3 \ \int_{0}^{\infty}  \int_{0}^{\infty}   P(\min(x,z)>X_3)^2 dF(z)dF(x) \\ \nonumber
    &=&   p_k^3 \int_{0}^{\infty} \left[ \int_0^x P(z>X_3)^2 dF(z)+ \int_x^{\infty} P(x>X_3)^2 dF(z) \right]dF(x) \\\nonumber
    &=&  p_k^3 \int_{0}^{\infty} \left[ \int_0^x F(z)^2 dF(z)+ \int_x^{\infty} F(x)^2 dF(z) \right]dF(x)\\ \nonumber
    &=&   p_k^3 \int_{0}^{\infty} \left[ \frac{F(x)^3}{3}+ F(x)^2 (1-{F}(x)) \right]dF(x)\\ \label{eq:terma1}
    &=& \frac{p_k^3}{6}.
 \end{eqnarray}
Again, under $H_0$, we have
\begin{eqnarray} \nonumber
    a_2&=& P(\min(X_1,X_2)>X_3, Y_1=Y_2=k)^2\nonumber\\\nonumber &=& \left(p_k^2 P(\min(X_1,X_2) > X_3)\right)^2\\ \label{eq:terma2}
     &=& \frac{p_k^4}{9}.
     \end{eqnarray}
   Thus, from equation \eqref{eq:terma1} and \eqref{eq:terma2}, we have
 \begin{equation}\label{termA:sigkk}
     A = \frac{p_k^3}{6}-\frac{p_k^4}{9}.
 \end{equation}

 Consider the term $B$ in equation \eqref{sigkk},
 \begin{eqnarray} \nonumber
     B &=& \Cov[I(\min(X_1,X_2)>X_3, Y_1=Y_2=k),I(\min(X_4,X_5)>X_1, Y_4=Y_5=k)]\\ \nonumber
     &=& E[I(\min(X_1,X_2)>X_3, Y_1=Y_2=k).I(\min(X_4,X_5)>X_1, Y_4=Y_5=k)]\\ \nonumber
     &-& P(\min(X_1,X_2) >X_3, Y_1=Y_2 =k) \times P(\min(X_4,X_5) >X_1, Y_4=Y_5 =k)\\ \nonumber
     &=& b_1 -b_2.
\end{eqnarray}

Consider the term $b_1$, under the null hypothesis
\begin{eqnarray} \nonumber
    b_1 &=& E[I(\min(X_1,X_2)>X_3, Y_1=Y_2=k).I(\min(X_4,X_5)>X_1, Y_4=Y_5=k)]\\ \nonumber
    &=&  \sum_{y=1}^{K}\int_{0}^{\infty}  E((\min(x, X_2)>X_3, \min(X_4,X_5) >x,
    Y_2=Y_4=Y_5 =y,y=k)\\\nonumber&&\qquad \qquad \qquad \qquad \qquad \qquad \qquad \qquad \qquad \qquad \qquad \qquad  P(Y_1 =y) dF(x)\\ \nonumber
     &=& p^4_k \int_{0}^{\infty} P(\min(x,X_2)>X_3) \overline{F}(x)^2 dF(x) \\ \nonumber
     &=& p^4_k \int_{0}^{\infty} \int_{0}^{\infty}P(\min(x,z)>X_3) \overline{F}(x)^2 dF(z) dF(x)\\ \nonumber
     &=& p_k^4 \int_{0}^{\infty} \Bigg[ \int_{0}^{x
     } P(\min(x,z)>X_3) \overline{F}(x)^2dF(z) \\\nonumber &&\qquad \qquad \qquad \qquad +\int_{x}^{\infty
     } P(\min(x,z)>X_3) \overline{F}(x)^2 dF(z)\Bigg] dF(x) \\ \nonumber
     &=&  p_k^4 \int_{0}^{\infty} \Bigg[ \int_{0}^{x
     } F(z) \overline{F}(x)^2dF(z) +\int_{x}^{\infty
     } F(x) \overline{F}(x)^2 dF(z)\Bigg] dF(x) \\ \nonumber
     &=&  p_k^4  \left[ \int_{0}^{\infty} \frac{F(x)^2 }{2}\overline{F}(x)^2dF(x) +\int_{0}^{\infty
     } F(x) \overline{F}(x)^3 dF(x)\right] \\ \nonumber
     &=& \frac{p_k^4}{2} \int_{0}^{\infty} (1+ F(x)^2 -2F(x) ) F(x)^2 dF(x) +  p_k^2 \int_{0}^{\infty} F(x) \overline{F}(x)^3 dF(x) \\ \nonumber
     &=& \frac{p_k^4}{2} \int_{0}^{\infty} (F(x)^2+ F(x)^4 -2F(x)^3 )  dF(x) +  p_k^2 \int_{0}^{\infty}  (\overline{F}(x)^3 -  \overline{F}(x)^4)dF(x) \\  \nonumber
     &=& \frac{p_k^4}{2} \left[ 1/3 +1/5 -2/4\right] +p_k^4 \left[ 1/4- 1/5\right]\\ \label{eq:termb1}
     & =&  \frac{p_k^4 }{15}.
\end{eqnarray}
Note that the term $b_2$ is similar to $a_2$. Thus, we have
\begin{equation} \label{eq:termb2}
    b_2 = \frac{p_k^4}{9}.
\end{equation}
 From equation \eqref{eq:termb1} and \eqref{eq:termb2} we have,
 \begin{eqnarray} \nonumber
   B &=& b_1 -b_2\\ \nonumber
   & =&  \frac{p_k^4 }{15} - \frac{p_k^4}{9}\\ \label{eq:termb}
   & =& - \frac{2}{45}p_k^4.
 \end{eqnarray}

  Consider the term $C$ in equation \eqref{sigkk},
\begin{eqnarray} \nonumber
    C &=& \Cov[I(\min(X_2,X_3)>X_1, Y_2=Y_3=k),I(\min(X_4,X_5)>X_1, Y_4=Y_5=k)]\\ \nonumber
    &=&  E[I(\min(X_2,X_3)>X_1, Y_2=Y_3=k),I(\min(X_4,X_5)>X_1, Y_4=Y_5=k)] \\ \nonumber
    &-& P(\min(X_2,X_3) > X_1 , Y_2=Y_3 =k)^2 \\ \nonumber
    &=& c_1 -c_2.
\end{eqnarray}

Consider the term $c_1$, under the null hypothesis
\begin{eqnarray} \nonumber
    c_1 &=&E[I(\min(X_2,X_3)>X_1, Y_2=Y_3=k),I(\min(X_4,X_5)>X_1, Y_4=Y_5=k)] \\ \nonumber
     &=&  \int_{0}^{\infty} P(\min(X_2,X_3) > x, Y_2 =Y_3 =k) P(\min(X_4,X_5) > x , Y_4 =Y_5 =k ) dF(x).
 \nonumber\\\nonumber
  &=& \int_{0}^{\infty}  p_k^4 P(\min(X_2,X_3) >x ) P(\min(X_4, X_5)>x) dF(x) \\ \nonumber
 &=&  \int_{0}^{\infty}  p_k^4 \overline{F}(x) ^4 dF(x) \\ \nonumber
    & =&  \int_{0}^ {\infty} p_k^4 \overline{F}(x) ^4  dF(x) \\ \label{eq:termc1}
    &=& \frac{p_k^4}{5}.
\end{eqnarray}

Note that term $c_2$ is similar to $a_2$, and hence, we have,
\begin{equation} \label{eq:termc}
    C = \frac{p^4_k}{5}-\frac{p_k^9}{9} = \frac{4}{45} p_k^4.
\end{equation}
From equation \eqref{termA:sigkk}, \eqref{eq:termb} and \eqref{eq:termc}, we have
\begin{equation} \label{eq:sigkk_fin}
    \sigma_{kk} = \frac{1}{2}p_k^3 - \frac{7}{15}p_k^4.
\end{equation}

Now, let us calculate the value of $\sigma_{kl}$:
\begin{eqnarray}\nonumber
\sigma_{kl}&=& \Cov[I(\min(X_1,X_2)>X_3, Y_1=Y_2=k)+I(\min(X_2,X_3)>X_1, Y_2=Y_3=)\\ \nonumber
&+&I(\min(X_1,X_3)>X_2, Y_1=Y_3=k),I(\min(X_1,X_4)>X_5, Y_1=Y_4=k)\\ \nonumber
&+&I(\min(X_4,X_5)>X_1, Y_4=Y_5=k)+I(\min(X_1,X_5)>X_4, Y_1=Y_5=k) ]\\ \nonumber
 &=& 3 \Cov[I(\min(X_1,X_2)>X_3, Y_1=Y_2=k),I(\min(X_1,X_4)>X_5, Y_1=Y_4=k)]\\ \nonumber
 &+& 5\Cov[I(\min(X_1,X_2)>X_3, Y_1=Y_2=k),I(\min(X_4,X_5)>X_1, Y_4=Y_5=k)]\\ \nonumber
 &+& \Cov[I(\min(X_2,X_3)>X_1, Y_2=Y_3=k),I(\min(X_4,X_5)>X_1, Y_4=Y_5=k)]\\ \label{sigkl}
 &=& 3A^*+5B^*+C^*.
\end{eqnarray}

Similar to the calculation of $\sigma_{kk}$, under $H_0$,  we show that term $A^*$ becomes
\begin{eqnarray} \nonumber
    A^{*} &=& \Cov[I(\min(X_1,X_2)>X_3, Y_1=Y_2=k),I(\min(X_1,X_4)>X_5, Y_1=Y_4=k)]\\ \label{eq:terma_s}
    &=& -\frac{p_k^2p_l^2}{9}.
\end{eqnarray}

Now, the term $B^{*}$ in equation \eqref{eq:sigkl} under $H_0$ becomes
\begin{eqnarray} \nonumber
    B^{*} &=& \Cov[I(\min(X_1,X_2)>X_3, Y_1=Y_2=k),I(\min(X_4,X_5)>X_1, Y_4=Y_5=k)] \\ \nonumber
    &=& \frac{p_k^2p_l^2}{15} - \frac{p_k^2p_l^2}{9} \\ \label{eq:termb_s}
    &=& -\frac{2}{45}p_k^2p_l^2.
\end{eqnarray}
And, the term $C^{*}$ in equation \eqref{eq:sigkl} under $H_0$ becomes
\begin{eqnarray} \nonumber
    C^{*} &=& \Cov[I(\min(X_2,X_3)>X_1, Y_2=Y_3=k),I(\min(X_4,X_5)>X_1, Y_4=Y_5=k)] \\ \nonumber
    &=& \frac{p_k^2p_l^2}{5} - \frac{p_k^2p_l^2}{9}\\ \label{eq:termc_s}
    &=& \frac{4}{45}p_k^2p_l^2.
\end{eqnarray}
Thus from equation \eqref{eq:terma_s}, \eqref{eq:termb_s} and \eqref{eq:termc_s}, we obtain
\begin{equation*}
    \sigma_{kl} = -\frac{7}{15} p_k^{2}p_l^{2}.
\end{equation*}
\end{document}